\documentclass[twocolumn, aps, pra, superscriptaddress, floatfix, nofootinbib, longbibliography, 10pt]{revtex4-2}
\usepackage{amsmath,amsfonts,amssymb,bm,graphicx,float,subcaption,xcolor,hyperref,caption,subcaption,csquotes,braket,amsthm,relsize}
\usepackage[mathscr]{euscript}
\usepackage{ulem} 
 \let\mathscr\relax
\usepackage[scr]{rsfso}

\newcommand{\PQD}{\hbox{PQD}}
\newcommand{\GSP}{\hbox{GSP}}
\newcommand{\sPQD}{\hbox{($s$)-PQD}}
\newcommand{\dm}[1]{d^{\,2}{#1}\,}

\newcommand{\Tr}{\text{Tr}}
\newtheorem*{theorem}{Theorem}

\begin{document}

\title{Filter functions for the Glauber-Sudarshan $P$-function regularization} 
\author{Mani Zartab}
\affiliation{Department of Physics, University of Tehran, P.O. Box 14395-547, Tehran, Iran}
\affiliation{Física Teòrica: Informació i Fenòmens Quàntics, Departament de Física,
	Universitat Autònoma de Barcelona, 08193 Bellaterra (Barcelona), Spain}
\author{Ezad Shojaee}
\affiliation{National Institute of Standards and Technology, Boulder, Colorado 80305, USA}
\affiliation{Department of Physics, University of Colorado, Boulder, Colorado, 80309, USA}
\author{Saleh Rahimi-Keshari}
\affiliation{Department of Physics, University of Tehran, P.O. Box 14395-547, Tehran, Iran}
\affiliation{School of Physics, Institute for Research in Fundamental Sciences (IPM), P.O. Box 19395-5531, Tehran, Iran}


\begin{abstract}
The phase-space quasi-probability distribution formalism for representing quantum states provides practical tools for various applications in quantum optics such as identifying the nonclassicality of quantum states. We study filter functions that are introduced to regularize the Glauber-Sudarshan $P$ function. We show that the quantum map associated with a filter function is completely positive and trace preserving and hence physically realizable if and only if the Fourier transform of this function is a probability density distribution. We also derive a lower bound on the fidelity between the input and output states of a physical quantum filtering map. Therefore, based on these results, we show that any quantum state can be approximated, to arbitrary accuracy, by a quantum state with a regular Glauber-Sudarshan $P$ function. We propose applications of our results for estimating the output state of an unknown quantum process and estimating the outcome probabilities of quantum measurements. 
\end{abstract}
\maketitle

\section{Introduction}
The well-developed theory of phase-space quasiprobability distributions (\PQD s) plays an indispensable role in representing the quantum states of optical systems with infinite-dimensional Hilbert spaces. In particular, by using the Glauber-Sudarshan representation~\cite{Glauber1963,Sudarshan1963}, the density operator $\rho$ of a single-mode system can be expressed in terms of coherent states $\ket{\alpha}$,
\begin{equation}
	\label{eq:P-representation}
	\rho=\!\int\!\dm{\alpha} P(\alpha) \ket{\alpha}\!\bra{\alpha},
\end{equation}
where $P(\alpha)$ is called the Glauber-Sudarshan $P$-function (\GSP) and the integration is over the entire complex plane. This representation is central in quantum optics, in particular, for defining nonclassical states, whose \GSP\ is not a probability density~\cite{Mandel1986}. Nonclassicality is an important quantum feature that is related to entanglement generation in linear-optical networks~\cite{Asboth2005,Sperling2014,Gholipour2016} and is a resource for quantum computation~\cite{Veitch2013,Mari2012,Rahimi2016}, quantum metrology~\cite{Kwon2019,Ge2020} and the incompatibility of quantum measurements~\cite{Rahimi2021}.

The Glauber-Sudarshan representation has been widely used in quantum optics. It provides an interesting way of calculating the expectation values of normally ordered operators, which are relevant in the context of photodetection and measuring correlation functions~\cite{mandel_wolf_1995,Vogel2006}, and it is also useful in describing the evolution of open quantum systems~\cite{Carmichael1999}. Moreover, this representation allows us to express the action of any quantum process on a quantum state, by using the linearity property, in terms of its action on coherent states, $\mathcal{E}(\rho)=\!\int\!\dm{\alpha} P(\alpha) \mathcal{E}(\ket{\alpha}\!\bra{\alpha})$. This relation is useful for quantum information applications, in particular, for probing unknown quantum processes using coherent states, which are readily available from a laser source~\cite{Lobino2008,Rahimi2011}. However, the main obstacle in using the Glauber-Sudarshan representation is that for most nonclassical states the \GSP\ is highly singular, existing only as a generalized distribution~\cite{Cahill1965,SperlingPRA2016}. Therefore, an interesting question is whether this problem can be obviated by approximating any state, to arbitrary accuracy, by another state with a regular \GSP. 

A general strategy for regularizing the \GSP\ of a quantum state is to multiply its Fourier transform with a filter function~\cite{Klauder1966}. This filtering procedure can be viewed as a map that transforms any density operator $\rho$ with a singular \GSP\ to an operator $\rho_\Omega=\mathcal{E}_{\Omega}(\rho)$ whose \GSP\ is regularized. Using a specific class of filter functions, Klauder showed that $\rho_\Omega$ can be close to $\rho$ within an arbitrary accuracy~\cite{Klauder1966}. However, as we show in this paper, the filtering map associated with Klauder's filter cannot be described as a physical process and hence the operator $\rho_\Omega$ may not be a physical state.

The physical filtering maps are important as they can, in principle, be realized and used in quantum experiments~\cite{Kuhn2018}. In addition, the physicality of a filtering map $\mathcal{E}_{\Omega}$ guarantees that the output state $\rho_\Omega$ is a physical density operator, and therefore the states before and after filtering can be compared using standard distance measures with operational interpretations in quantum information, such as the fidelity and the trace distance~\cite{NielsenChuang}.

Nonclassicality filters are another class of filter functions that are proposed to verify the nonclassicality of quantum states~\cite{Kiesel2010,Kuhn2014}. The filtering maps associated with nonclassicality filters are physical~\cite{Kuhn2018}; they preserve the classicality of quantum states and convert the \GSP\ of nonclassical states into a regular function with some negativity. However, to use these filter functions for quantum state approximation, it is important to bound the distance between states before and after filtering.

In this paper, we identify the general class of filters that can be described by completely positive trace-preserving (CPTP) linear maps, which are hence physically implementable. Specifically, we show that the filtering map is CPTP if and only if the filter function is the Fourier transform of a probability density distribution. Thus, following the method in~\cite{Kuhn2018}, the filtering can be realized by a random application of the displacement operator on the initial state according to the Fourier transform of the filter function. We apply this condition to several examples of filter functions, and also identify a class of positive but not completely positive maps. In addition, we derive a lower bound on the fidelity between the input and output states of CPTP filters. Therefore, for an arbitrary bound on the accuracy, we can approximate any nonclassical state with a highly singular \GSP\ with a quantum state whose \GSP\ is regular and can be expressed more easily in the Glauber-Sudarshan representation. We discuss some applications of this formalism in quantum information processing, such as estimating the output of a quantum channel and estimating the outcome probabilities of a general measurement given the heterodyne record.   
	
The outline of this paper is as follows. We start by reviewing the general phase-space formalism and the filtering procedure in Sec.~\ref{sec:QPD-Filter}. In Sec.~\ref{sec:condition}, we introduce the necessary and sufficient condition for a filtering map, induced by the filter function, to be a physical quantum process, i.e., a CPTP map. Examples of various filter functions are discussed in Sec.~\ref{sec:examples}. Then we derive a lower bound on the fidelity between the filtered and unfiltered states in Sec.~\ref{sec:bound}. We discuss two applications of physical filters for regularizing the \GSP\ in Sec.~\ref{sec:application} and conclude the paper in Sec.~\ref{sec:conclusion}.

\section{Phase-space quasiprobability distributions and filter functions}\label{sec:QPD-Filter}

A density operator describing the physical state of a single-mode bosonic system can be expressed as~\cite{Cahill1969-OrderedExp}
\begin{equation}\label{eq:charac-expan}
	\rho= \frac{1}{\pi}\!\int\!\dm{\xi} \Phi(\xi)\, D(-\xi)
\end{equation}
where the integral is over the entire complex plane, $D(\xi)=\exp(\xi a^{\dagger}-\xi^*a)$ is the displacement operator, with $a^{\dagger}$ and $a$ being the creation and annihilation operators, respectively; and $\Phi(\xi)=\Tr[\rho D(\xi)]$ is the characteristic function\footnote{Here, $\hbar=1$; hence, $a=\frac{1}{\sqrt{2}}(\hat{x}+i\hat{p})$ in terms of canonical position and momentum operators, and $\xi=\frac{1}{\sqrt{2}}(u+iv)$ in terms of canonical coordinates.}. The physicality conditions of density operators ($\rho\ge 0$ and $\Tr[\rho]=1$) translates into the following conditions for physical characteristic functions: $\Phi(\xi)$ is continuous, $\Phi(0)=1$, and for every finite set of points $\{\xi_i\}$ the matrix $M_{jk}=\Phi(\xi_j-\xi_k) \exp[(\xi_j\xi_k^*-\xi_j^*\xi_k)/2]$ is non-negative~\cite{Kastler1965,Loupias1966,Narcowich1986,Nha2008}. \\

Equation~(\ref{eq:charac-expan}) can also be formulated as 
\begin{align}\label{eq:s-repres}
	\rho=\!\int\!\dm{\alpha} W^{(s)}(\alpha)\, T^{(-s)}(\alpha),
\end{align}
where the operators $T^{(-s)}(\alpha)$ are given by
\begin{align}
T^{(-s)}(\alpha)= \frac{1}{\pi}\!\int\!\dm{\xi} e^{-s|\xi|^2/2} D(-\xi)\, e^{\alpha\xi^*-\xi\alpha^*},
\end{align}
and
\begin{align}\label{eq:spqd-def}
	W^{(s)}(\alpha)=\frac{1}{\pi^2}\!\int\!\dm{\xi} \Phi(\xi)\, e^{s|\xi|^2/2} e^{\alpha\xi^*-\xi\alpha^*}
\end{align}
is known as the $s$-ordered phase-space quasiprobability distributions [\sPQD s]~\cite{Cahill1969-DenOper,HILLERY1984}. Here, $s$ is the ordering parameter, and for $s=1$, Eq.~(\ref{eq:s-repres}) becomes the Glauber-Sudarshan representation with $T^{(-1)}(\alpha)=\ket{\alpha}\!\bra{\alpha}$. For $s=0$, the \sPQD\ is the Wigner function~\cite{Wigner}, which is a regular function that may take on negative values for some nonclassical states. For $s=-1$, the \sPQD\ is the Husimi $Q$ function~\cite{Husimi}, which is always non-negative for positive operators. Note that for all values of  $s<1$ the operator $T^{(-s)}(\alpha)$ does not correspond to a physical density operator, and for $0<s\leq1$ the \sPQD\ can be highly-singular. 

Because of the many applications that phase-space representation offers, it is very useful to tame the singularities arising in \sPQD s, in particular, in the case of $s=1$ for the Glauber-Sudarshan representation. A general way to do that is to introduce a filter function that can regularize highly singular \sPQD s~\cite{Klauder1966}. Specifically, the filtering procedure is defined by multiplying the characteristic function with a filter function $\Omega(\xi)$,
\begin{equation} \label{eq:filtering-characteristic}
	\Phi_{\Omega}(\xi)=\Phi(\xi)\,\Omega(\xi).
\end{equation}
In this case, $\Phi_{\Omega}(\xi)$ can be thought of as the characteristic function of another operator $\rho_\Omega$ that is formally given by a linear map $\mathcal{E}_{\Omega}$,
\begin{align}\label{eq:filterng-map}
	\rho_\Omega=\mathcal{E}_\Omega(\rho)=\!\int\!\dm{\alpha} \tilde{\Omega}(\alpha) D(\alpha)\rho D^{\dagger}(\alpha),
\end{align}
where $\tilde{\Omega}(\alpha)$ is the Fourier transform of the filter function. One can verify this equation by using Eq.~(\ref{eq:charac-expan}) and $D(\alpha)D(-\xi)D^\dagger(\alpha)=\exp(\xi\alpha^*-\alpha\xi^*)D(-\xi)$. Notice that $\rho_\Omega$ may not necessarily be a density operator, as we discuss in the next section. 
By taking the Fourier transformation of Eq.~(\ref{eq:filtering-characteristic}) the \sPQD\ of operator $\rho_\Omega$ reads
\begin{align}\label{eq:spqd-convolution}
	\begin{split}
	W^{(s)}_{\Omega}(\alpha)=W^{(s)}\!\ast\tilde{\Omega}(\alpha)
	=\!\int\!\dm{\beta} W^{(s)}(\alpha-\beta)\;\tilde{\Omega}(\beta),
	\end{split}
\end{align}
which is the convolution of the \sPQD\ of density operator $\rho$ and Fourier transform of the filter function. Using this, the output operator of the filtering map, Eq.~(\ref{eq:filterng-map}) is given by 
\begin{equation}
	\rho_\Omega=\!\int\!\dm{\alpha} W^{(s)}_{\Omega}(\alpha)\, T^{(-s)}(\alpha).
\end{equation}
This shows that the filtering map, in general, can be described in terms of the \sPQD s. However, of our particular interest in this paper is the case $s=1$, corresponding to the Glauber-Sudarshan representation (\ref{eq:P-representation}), where Eq.~(\ref{eq:spqd-convolution}) becomes $P^{(s)}_{\Omega}(\alpha)=P^{(s)}\!\ast\tilde{\Omega}(\alpha)$. 

\section{Necessary and sufficient condition for physicality of a filtering map}\label{sec:condition}

A quantum process, described by a CPTP map, transforms density operators to density operators and, in principle, can be realized in the laboratory~\cite{NielsenChuang}. However, in general, the filtering map~(\ref{eq:filterng-map}) generated by the filter function $\Omega(\xi)$ may not be a physical quantum process. Here, we lay out the necessary and sufficient condition for a filter function to be CPTP, hence preserving the physicality of quantum states. 

\begin{theorem}
A filtering map $\mathcal{E}_\Omega$ is CPTP \textit{iff}  the filter function $\Omega(\xi)$ is the Fourier transform of a probability density function.
\end{theorem}

\begin{proof}
If the Fourier transform of the filter function, $\tilde{\Omega}(\alpha)$, is a probability density, then according to Eq.~(\ref{eq:filterng-map}), ${\mathcal{E}}_\Omega(\rho)$ is a statistical mixture of displaced density operators $D(\alpha)\rho D^{\dagger}(\alpha)$, a valid density operator. Therefore, $\tilde{\Omega}(\alpha)$ being a probability density function is sufficient for $\mathcal{E}_\Omega$ to be a CPTP map. 

To prove that this condition is also necessary, suppose that our system is a subsystem $S$ of a bipartite system in the joint state $\rho_{SE}$. If ${\mathcal{E}}_\Omega$ is completely positive, then the state of the joint system after applying the map to subsystem $S$ and the identity $\mathcal{I}$ on subsystem $E$,  $(\mathcal{E}_\Omega\otimes\mathcal{I})\rho_{SE}$, must remain a positive operator for any $\rho_{SE}$, where the other subsystem $E$ can be any arbitrary quantum system. However, if $\tilde{\Omega}(\alpha)$ takes on negative values, meaning if it is not a probability density, we show that $(\mathcal{E}_\Omega\otimes\mathcal{I})\rho_{SE} \geq 0$ does not hold by a counterexample. Consider a bipartite bosonic system in a two-mode squeezed vacuum state
\begin{equation}
\ket{\psi_{SE}}= \sqrt{1-\chi^2}\sum_{n=0}^\infty \chi^n \ket{n}_{S}\otimes\ket{n}_{E},
\end{equation}
where $\ket{n}$ are the number states and $0\leq\chi<1$ is the squeezing parameter. By using Eq.~(\ref{eq:filterng-map}), we calculate the fidelity between the output operator $\sigma_{SE}=(\mathcal{E}_\Omega\otimes\mathcal{I})\ket{\psi_{SE}}\!\bra{\psi_{SE}}$ and the state $\ket{\phi_{SE}}=D(\beta)\otimes\mathcal{I}\ket{\psi_{SE}}$
\begin{equation}\label{eq:fidl-filter-map}
	\begin{split}
\bra{\phi_{SE}}\sigma_{SE}&\ket{\phi_{SE}}
=\!\int\!\dm{\alpha} \tilde{\Omega}(\alpha)\\ &\times\big|\!\bra{\psi_{SE}}\left(D(\alpha-\beta)\otimes\mathcal{I}\right)\ket{\psi_{SE}}\!\big|^2.
 \end{split}
\end{equation} 
Having the identity $\bra{n}D(\gamma)\ket{n}=\exp(-|\gamma|^2/2) L_n(|\gamma|^2)$ with $L_n(x)$ being the Laguerre polynomial of the order $n$ and its generating function $\sum_{n=0}^\infty t^n L_n(x)=1/(1-t)\exp\left(-tx/(1-t)\right)$, we obtain
\begin{align}
\bra{\psi_{SE}}\!\big[D(\gamma)\otimes\mathcal{I}\big]\!\ket{\psi_{SE}}&=(1-\chi^2)\sum_{n=0}^\infty \chi^{2n} \bra{n}D(\gamma)\ket{n}\nonumber\\
&=\exp\!\left(\!-\frac{|\gamma|^2}{2}-\frac{\chi^2|\gamma|^2 }{1-\chi^2}\right).
\end{align}
Therefore, the fidelity~(\ref{eq:fidl-filter-map}) becomes
\begin{equation}\label{eq:fidelity-filter-exp}
\bra{\phi_{SE}}\sigma_{SE}\ket{\phi_{SE}}
=\int\!\dm{\alpha} \tilde{\Omega}(\alpha)\, e^{-(1+2\bar{n})|\alpha-\beta|^2},
\end{equation}
where $\bar{n}=\chi^2/(1-\chi^2)$ is the mean photon number of the reduced density operator (traced over $E$) of $\ket{\psi_{SE}}$ that is a thermal state. By noting that the Dirac delta function can be defined as the limit of a Gaussian function,
$\delta^2(\gamma)=\delta(\operatorname{Re}(\gamma))\delta(\operatorname{Im}(\gamma))=\lim_{\bar{n}\to\infty}(2\bar{n}/\pi)\exp(-2\bar{n}|\gamma|^2)$, we see that if $\tilde{\Omega}(\alpha_0)<0$ for some $\alpha_0$, by choosing $\beta=\alpha_0$ in Eq.~(\ref{eq:fidelity-filter-exp}), we get
\begin{equation}
	\lim_{\bar{n} \to \infty}\frac{2\bar{n}}{\pi}\bra{\phi_{SE}}\sigma_{SE}\ket{\phi_{SE}}=\tilde{\Omega}(\alpha_0)<0.
\end{equation}  
This implies that by choosing a sufficiently large value of $\bar{n}$ the fidelity $\bra{\phi_{SE}}\sigma_{SE}\ket{\phi_{SE}}$ becomes negative, indicating that $\sigma_{SE}$ is not a positive operator. Therefore, if the Fourier transform of the filter function takes on negative values, the filtering map does not preserve the physicality of the entangled state $\ket{\psi_{SE}}$ and is not completely positive. Also, for the filtering map $\mathcal{E}_\Omega$ to be trace preserving, $\tilde{\Omega}(\alpha)$ must also be normalized, as $\Tr[\mathcal{E}_\Omega(\rho)]\!=\!\int\dm{\alpha} \tilde{\Omega}(\alpha)\!=\!\Omega(0)\!=\!1$.
Therefore, the filtering map is CPTP if and only if the function $\tilde{\Omega}(\alpha)$ is a probability density distribution.
\end{proof}

If this condition is satisfied, as we can see from Eq.~(\ref{eq:filterng-map}), the filtering process can be realized by applying a displacement operator chosen randomly according to the probability density distribution $\tilde{\Omega}(\alpha)$. In practice, to implement a displacement operation on a quantum state, one can overlap the state on a highly transmissive beam splitter with a coherent state~\cite{Kuhn2018}. So in this case, the filtered state $\rho_{\Omega}$ is always mixed and can be viewed as a noisy version of the original state $\rho$.

\section{Examples of filter functions}\label{sec:examples}

In this section, we consider several examples of filter functions that were previously studied in the literature and check whether the associated filtering maps are CPTP or not.

\subsection{Gaussian filters}
One simple example of filter functions corresponding to CPTP maps is the Gaussian function $\Omega_{r}(\xi)=\exp(-r|\xi|^2/2)$, where $r$ is a positive number. The Fourier transform is a Gaussian function $\tilde{\Omega}(\alpha)=2\exp(-2|\alpha|^2/r)/(r\pi)$, which is already a probability density. In fact, by using Eq.~(\ref{eq:spqd-def}), we can see that this is the filter function that relates the \sPQD s of the density operator in terms of the convolution (\ref{eq:spqd-convolution}),
\begin{equation}\label{eq:Gaus-filter}
	W^{(s)}_{\Omega}(\alpha)=W^{(s-r)}(\alpha)= W^{(s)}\!\ast \frac{2}{\pi r} \exp\!\bigg(\!\frac{-2|\alpha|^2}{r}\bigg).
\end{equation}
Note that for $1\leq r\leq 2$, the filtered \sPQD\ corresponds to quasiprobability distributions between the Wigner function and the Husimi $Q$ function, which are regular functions.

We can make some observations based on Eq.~(\ref{eq:Gaus-filter}) for $r \geq 0$. First, the \hbox{($s{-}r$)-PQD}\ of a density operator $\rho$ can be thought of as the \sPQD\ of a filtered density operator $\rho_{\Omega}$, which is obtained by applying random displacements to $\rho$ according to the Gaussian probability density $\tilde{\Omega}(\alpha)$. For example, by setting $s=r$, this implies that the \hbox{($0$)-PQD}, i.e., the Wigner function of a quantum state, can be viewed as the \hbox{($r$)-PQD}\ of the filtered state. The second observation is that a density operator $\rho$ is mixed if its \sPQD\ is equal to the \hbox{($s{-}r$)-PQD}\ of another density operator; this is because $\rho$ can be thought of as the output of the Gaussian filtering process with parameter $r$, which is always mixed. For instance, we can immediately see that $\rho$ is mixed if its \GSP\ is equal to the Wigner function of a number state. The last observation is that the \sPQD\ of a pure state cannot be the \hbox{($s{-}r$)-PQD}\ of any other physical state because the latter is equivalent to the \sPQD\ of a mixed filtered state $\rho_{\Omega}$.  As an interesting example, setting $s=r=1$ implies that the Wigner function cannot be a Dirac delta function that is the \GSP\ of a coherent state.

\subsection{Nonclassicality filters}\label{subsec:noncl-filt}

As another class of filter functions that can be described by CPTP maps, we consider the nonclassicality filter functions that regularize singular \GSP s of all nonclassical states. A nonclassicality filter can be defined as the autocorrelation of an infinitely differentiable function $\omega_L(\xi)$ \cite{Kiesel2010},
\begin{align}\label{eq:auto-correlation}
	\Omega_{L}(\xi)=\!\int\!d^2\zeta\, \omega_L(\zeta)\,\omega_L(\xi-\zeta),
\end{align}
where $0\leq L\leq\infty$ is the width parameter such that $\lim_{L \to \infty}\Omega_{L}(\xi)=\Omega_{L}(0)=1$. Therefore, the Fourier transform of the filter function $\tilde\Omega_L(\xi)$ is always guaranteed to be a probability density, and the negativity in the regularized \GSP $P_{\Omega}(\alpha)$, which is the \GSP\ of $\rho_{\Omega}$~\cite{Kuhn2018} and can be experimentally measured, indicates the nonclassicality of the state $\rho$. The nonclassicality can also be verified by measuring the moments of $\rho_{\Omega}$~\cite{RahimiPRA2012}.   

A class of nonclassicality filters is defined using 
\begin{align}\label{eq:omega_L}
	\omega_L(\xi) =\frac{1}{L} 2^{{1}/{q}}\sqrt{\frac{q}{2\pi\Gamma({2}/{q})}}\exp\!\bigg(\!\!-\frac{|\xi|^q}{L^q}\bigg),
\end{align}
where $2<q<\infty$ is a parameter characterizing the analytic form of $\omega_L(\xi)$, and $\Gamma$ is the gamma function~\cite{Kuhn2014}. The case $q=2$ corresponds to the Gaussian filters, which do not decay fast enough to regularize the \GSP\ of all nonclassical states~\cite{Kiesel2010}. Hence, the nonclassicality filters are non-Gaussian.

Nonclassicality filters have been used to experimentally verify the nonclassicality of single-photon-added thermal states~\cite{KieselPRA2011}, squeezed states of light~\cite{KieselPRL2011,AgudeloPRA2015}, an atomic spin-squeezed state~\cite{KieselPRA2012}, and single- and two-photon states using balanced homodyne detection~\cite{KuhnPRL2021}. These filters are also used to define nonclassicality witnesses~\cite{KieselPRA2012,KieselPRA2012-PNRD,KuhnPRL2016}, and they can be used to verify the nonclassicality of quantum processes, such as the single-photon addition process~\cite{RahimiPRL2013}.

\subsection{Klauder's filter}\label{Klaud-filter}

An example of a filter function whose filtering map is not CPTP is the one introduced by Klauder to approximate an arbitrary density operator by a bounded operator with an infinity differentiable \GSP \cite{Klauder1966}. We show that the associated filtering map is not physical and therefore the bounded operators obtained by this particular filtering map may not be valid density operators. Klauder's filter is defined as
\begin{align}
	\Omega^K_L(u,v)=e^{-[f(u-L)+f(-u-L)+f(v-L)+f(-v-L)]}
\end{align}
where $(u,v)=\sqrt{2}\left(\operatorname{Re}(\xi),\operatorname{Im}(\xi)\right)$ and $f(x)$ is defined as
\begin{align}
	f(x)=
	\begin{cases}
		x^4e^{-{1}/{x^2}} \quad \quad & x> 0 \\
		0 & x\leq0.
	\end{cases}
\end{align}
This filter function decays rapidly such that it regularizes the \GSP\ of all quantum states.

According to Bochner's theorem, the Fourier transform of the filter function $\tilde{\Omega}(\alpha)$ is probability density if and only if for any finite set of points $\{\xi_i\}$ the matrix $F_{jk}=\Omega(\xi_j-\xi_k)$ is positive semidefinite. Using the set of points $\{u_1=-L, u_2=0, u_3=L\}$ on the real axis, we can see that $\Omega^K_L(u_j-u_k,0)$ is not positive semidefinite because its determinant is negative,
\begin{align}
	\begin{split}
		\det\!\begin{pmatrix}
			1 & 1 & \Omega^K_L(-2L,0) \\
			1 & 1 & 1 \\
			\Omega^K_L(2L,0) & 1 & 1
		\end{pmatrix}\!=-(\Omega^K_L(2L,0)-1)^2,
	\end{split}\nonumber
\end{align}
where we used $\Omega^K_L(-2L,0)=\Omega^K_L(2L,0)$. Therefore, the Fourier transform of Klauder's filter, $\tilde\Omega^K_L(\xi)$, is not a probability density and the associated filtering map is not CPTP.

Klauder's filter function was introduced as a means to show that one can still use the Glauber-Sudarshan representation, despite the highly singular nature of the \GSP , to approximate the expectation value of any bounded operator with respect to a quantum state. Specifically, Klauder showed that one can always find a sequence of bounded operators, whose \GSP s are regularized by this filter function, that converge to the desired density operator in the trace-class norm~\cite{Klauder1966}. Klauder's filter function was later used for quantum process tomography using coherent states~\cite{Lobino2008}. In this method, a quantum process is represented by a rank-$4$ tensor, $\mathcal{E}_{nm}^{jk}=\bra{j}\mathcal{E}(\ket{n}\!\bra{m})\ket{k}$, that relates the matrix elements of the input and output density matrices in the Fock basis, $[\mathcal{E}(\rho)]_{jk}=\sum_{n m}\mathcal{E}_{nm}^{jk}\rho_{nm}$. Then, by using a regularized \GSP\ of $\ket{n}\!\bra{m}$ obtained with Klauder's filter $P_{\Omega,nm}(\alpha)$ and measuring the output density matrix for input coherent states $\bra{j}\mathcal{E}(\ket{\alpha}\!\bra{\alpha})\ket{k}$, the process tensor is approximated as $\mathcal{E}_{nm}^{jk}\approx\int d^2\alpha P_{\Omega,nm}(\alpha) \bra{j}\mathcal{E}(\ket{\alpha}\!\bra{\alpha})\ket{k}$. However, as we have shown, the filtering map associated with Klauder's filter does not necessarily preserve the physicality of quantum states, and this may lead to a nonphysical process tensor. The need to use a filter function was eliminated in an improved version of coherent-state quantum process tomography~\cite{Rahimi2011}. It was based on the fact that the \GSP\ of operators  $\ket{n}\!\bra{m}$ contains finite derivatives of the Dirac-$\delta$ function, and singularities of this type can be handled if one truncates the Hilbert space in the Fock basis. In this case, by measuring a finite number of samples from $\bra{j}\mathcal{E}(\ket{\alpha}\!\bra{\alpha})\ket{k}$ for various $\alpha$, the process tensor can be estimated~\cite{Rahimi2011}. Also, by considering an energy cutoff, the error introduced by the truncation of the Hilbert space can be upper bounded~\cite{Rahimi2011}.

\subsection{Wigner smoothing filters}

The marginals of the Wigner function are probability densities associated with the phase-space quadratures~\cite{Vogel2006}, but the Wigner function itself, in general, is not pointwise non-negative and hence cannot be viewed as a classical-like distribution over the phase space. This, in fact, reflects a fundamental distinction between quantum and classical phase-space theories. In this regard, methods for converting the Wigner function into non-negative, classical-like distributions in phase space, known as smoothed Wigner functions, have been of particular interest. Indeed, the Husimi $Q$ of a quantum state is one of those distributions that, according to Eq.~(\ref{eq:Gaus-filter}) for $r=1$ and $s=0$, can be viewed as the convolution of the Wigner function of the state and the Wigner function of the vacuum state. However, in a more general context of smoothing, one can convolve the Wigner function of the quantum state with the Wigner function of another state, known as the smoothing kernel~\cite{Jagannathan1987},
\begin{align}\label{eq:WDF smoothing}
	W_{\text{sm}}(\alpha)= W\!\ast W_{\text{ker}}(\alpha).
\end{align}
By comparing with Eq.~(\ref{eq:spqd-convolution}), we can see that the characteristic function of the smoothed Wigner function is given by $\Phi_{\text{sm}}(\xi)=\Phi(\xi)\Phi_{\text{ker}}(\xi)$, and the smoothing procedure is, in fact, a filtering map associated with the filter function $\Omega(\xi)=\Phi_{\text{ker}}(\xi)$, which is the characteristic function of the smoothing kernel. It was shown that $W_{\text{sm}}(\alpha)$ is always pointwise non-negative~\cite{Jagannathan1987}. Furthermore, if the smoothing kernel is a non-negative Wigner function, then the smoothing procedure corresponds to a CPTP map~\cite{Jagannathan1987,Narcowich1988}. 

Interestingly, in this context, we can identify filtering maps that are positive, i.e., preserve the physicality of single-mode states, but are not completely positive. For this purpose, it is necessary to introduce the concept of the \textit{Narcowich-Wigner spectrum}. Consider a function $f(\xi)$; the set of all real parameters $\eta$ for which the matrix with elements
\begin{align}
	F_{jk}=f(\xi_j-\xi_k)\exp\! \left({\frac{\eta}{4}(\xi_j\xi_k^*-\xi_j^*\xi_k)}\right)
\end{align}
is positive semidefinite for any finite set of complex numbers $\{\xi_i\}$ is called the Narcowich-Wigner spectrum of $f(\xi)$ and is denoted by $\mathcal{W}\{f(\xi)\}$. Based on this definition, one can verify the following properties: if $\eta \in \mathcal{W}\{f(\xi)\}$, then $ -\eta \in \mathcal{W}\{f(\xi)\}$; if $\eta_1 \in \mathcal{W}\{f(\xi)\}$ and $\eta_2 \in \mathcal{W}\{g(\xi)\}$, then $\eta_1+\eta_2 \in   \mathcal{W}\{f(\xi)g(\xi)\}$~\cite{Narcowich1988}. 

As discussed in Sec.~\ref{sec:QPD-Filter}, the characteristic function of a density operator must be continuous, $\Phi(0)=1$, and  $2 \in \mathcal{W}\{\Phi(\xi)\}$. In addition, based on Bochner's theorem, our necessary and sufficient condition for a continuous filter function $\Omega(\xi)$ with $\Omega(0)=1$ to generate a CPTP filtering map implies that $0 \in \mathcal{W}\{\Omega(\xi)\}$.
By using the above properties, one can show that if either $\{0,2\} \in \mathcal{W}\{\Phi_{\text{ker}}(\xi)\}$ or $\{2, 4\} \in \mathcal{W}\{\Phi_{\text{ker}}(\xi)\}$ is satisfied, then we always have  $\{0,2\} \in \mathcal{W}\{\Phi_{\text{sm}}(\xi)\}$, and therefore, $W_{\text{sm}}(\alpha)$ is a physical and pointwise nonnegative Wigner function ~\cite{Narcowich1988}. However, as we have shown, if $0 \notin \mathcal{W}\{\Phi_{\text{ker}}(\xi)\}$, then the filtering map is not completely positive. Thus, we can identify a class of filter functions whose filtering maps are positive but not completely positive: filter functions that have element 4 but not element 0 in their Narcowich-Wigner spectrum. An example of such a filter function is~\cite{Narcowich1988}
\begin{align}
	\Phi_\text{ker}(\xi)=\bigg(\!1-\frac{3}{2}|\xi|^2\!\bigg)e^{-|\xi|^2},
\end{align}
whose Fourier transform takes on negative values.

\section{Bound on the distance between original and filtered states}\label{sec:bound}

Having identified physical filtering processes described by CPTP maps, the question is now how to compare the states before and after the filtering process. In this section, we derive a lower bound on the fidelity between these two states, which are the output and input states of the CPTP filtering map $\mathcal{E}_\Omega$.

The fidelity between two states is given by $F(\rho_1,\rho_2)=\max\lvert\langle \mu_1| \mu_2 \rangle\rvert^2$, where the maximum is taken over all possible purifications $|\mu_1\rangle$ and $|\mu_2\rangle$ of $\rho_1$ and $\rho_2$, respectively~\cite{NielsenChuang}. Suppose $\ket{\mu_{SE}}\!=\!\sum_j \sqrt{\lambda_j} \ket{s_j}_{S}\!\otimes\!\ket{e_j}_{E}$ is a purification of $\rho$, where $\lambda_j$ are its eigenvalues and $\ket{s_j}_S$ and $\ket{e_j}_{E}$ are orthogonal bases for subsystems $S$ and $E$, respectively. Using this state and the filtering map $\mathcal{E}_\Omega$ given in Eq.~(\ref{eq:filterng-map}), we can calculate the entanglement fidelity~\cite{Schumacher1996,Nielsen1996}
\begin{align}\label{eq:ent-fild}
F_e(\rho,\mathcal{E}_\Omega)&=\bra{\mu_{SE}} \left(\mathcal{E}_\Omega\otimes\mathcal{I}\right)\left(\ket{\mu_{SE}}\!\bra{\mu_{SE}}\right)\ket{\mu_{SE}}\nonumber\\
&=\int\!\dm{\alpha} \tilde{\Omega}_L(\alpha)\, \big|\!\bra{\mu_{SE}}\!\big(D(\alpha)\otimes\mathcal{I}\big)\!\ket{\mu_{SE}}\!\big|^2\nonumber\\
&=\int\!\dm{\alpha}\tilde{\Omega}_L(\alpha)\,\lvert\Phi(\alpha)\rvert^2,
\end{align}
where in the last line we have used the definition of the characteristic function of $\rho$, $\bra{\mu_{SE}}\!\big(D(\alpha)\otimes\mathcal{I}\big)\!\ket{\mu_{SE}}=\Tr[\rho D(\alpha)]=\Phi(\alpha)$. Here, the width parameter $L$ of the filter function can be defined such that $\lim_{L \to \infty} \tilde{\Omega}_L(\alpha)=\delta^2(\alpha)$, as in the nonclassicality filters given by Eqs.~(\ref{eq:auto-correlation}) and (\ref{eq:omega_L}), for example. In this case, by choosing sufficiently large values of $L$ the entanglement fidelity~(\ref{eq:ent-fild}) can be made arbitrarily close to $\lvert\Phi(0)\rvert^2=1$.  

The entanglement fidelity is always a lower bound on the fidelity between the states before and after a quantum process~\cite{Nielsen1996}. Therefore, Eq.~(\ref{eq:ent-fild}) provides a lower bound on the fidelity between the filtered and unfiltered states, $F_e(\rho,\mathcal{E}_\Omega)\leq F(\rho,\rho_\Omega)$, which can be adjusted by using the width parameter $L$. For any $\epsilon >0$ one can choose $L$ such that 
\begin{equation} \label{fidelity}
	1-\epsilon\leq\!\int\!\dm{\alpha}\tilde{\Omega}_L(\alpha) \lvert\Phi(\alpha)\rvert^2\leq F(\rho,\rho_\Omega).
\end{equation}
This relation provides an upper bound on the error associated with the regularization of the \GSP\ or, in other words, using $\rho_\Omega$ instead of $\rho$. Therefore, this bound makes physical filtering maps a useful tool in quantum information applications, as we discuss in the following section.

If the unfiltered state is pure, $\rho=\ket{\psi}\!\bra{\psi}$, we have an exact expression for the fidelity 
\begin{align}
F(\rho,\rho_\Omega)\!=\!\bra{\psi}\!\mathcal{E}_\Omega\!\left(\ket{\psi}\!\bra{\psi}\right)\!\ket{\psi}\!=\!\!\int\!\!\dm{\alpha}\tilde{\Omega}_L(\alpha) \lvert\Phi(\alpha)\rvert^2\!,
\end{align}
which can be simply verified using Eq.~(\ref{eq:filterng-map}). Note also that if a CPTP filtering map is applied to a subsystem of a joint system, by using Eq.~(\ref{eq:ent-fild}), one can still find a lower bound for the fidelity between the joint states before and after the filtering map.

We also know that the fidelity between two states $\rho_1$ and $\rho_2$ provides an upper bound on the trace distance, $D(\rho_1,\rho_2)\leq \sqrt{1-F({\rho_1},\rho_2)}$, where $D(\rho_1,\rho_2)=\frac{1}{2}\Tr\left[|\rho_1-\rho_2|\right]$ with $|A|\equiv\sqrt{A^\dagger A}$~\cite{NielsenChuang}.   
Hence, the lower bound on the fidelity Eq.~(\ref{eq:ent-fild}) can also be used to obtain an upper bound on the trace distance between filtered and unfiltered states, 
\begin{align} \label{eq:tra-distance}
	\begin{split}
D(\rho,\rho_\Omega)\leq \sqrt{1-F_e({\rho},{\mathcal{E}}_\Omega)},
	\end{split}
\end{align}
where, as discussed before, this bound can be made arbitrarily small by choosing a sufficiently large value of the width parameter $L$.

\section{The \GSP\ regularization and its applications}\label{sec:application}
Our results for phase-space filtering maps have been, so far, general and can be applied to the \sPQD\ representation~(\ref{eq:s-repres}) for any value of $s$. We turn now to our main interest in this paper, the Glauber-Sudarshan representation~(\ref{eq:P-representation}), and consider physicality preserving filter functions ${\Omega}_L(\alpha)$ with a width parameter $L$ [$\lim_{L \to \infty}\Omega_L(\alpha)=1$] that can regularize the \GSP\ of quantum states. An example of such a filter is the nonclassicality filter, discussed in Sec.~\ref{subsec:noncl-filt}.  Based on this, an important implication of our results is that any quantum state $\rho$ can be approximated by another state,
\begin{equation}
\label{eq:reg-P-representation}
\rho_{\Omega}=\!\int\!\dm{\alpha} P_{\Omega}(\alpha) \ket{\alpha}\!\bra{\alpha},
\end{equation}
where $P_{\Omega}(\alpha)$ is a regular function and hence the integration can be performed in the usual way. The bound on the approximation error is given by Eqs.~(\ref{fidelity}) and (\ref{eq:tra-distance}) and can be tuned by varying the width parameter $L$. In the following, we discuss two interesting applications of this result for estimating the output of a quantum process using coherent states as a probe and estimating the outcome probability distribution of a general quantum measurement using heterodyne measurement.    

\subsection{Estimating the output of quantum processes with coherent states}

Consider a known state $\rho$ and an unknown CPTP process $\mathcal{E}$ for which we are interested in estimating the state $\mathcal{E}(\rho)$, i.e., the output of the channel. 
In standard quantum process tomography~\cite{NielsenChuang,Mohseni2008}, one sends a set of probe states to an unknown process, and the output states are measured. Using the effect of the process on the probe states, one can predict the output state for any input state within the same Hilbert space. This is done by considering a complete set of operators $\{\varrho_j\}$ that provides a basis for representing any input state $\rho=\sum_j c_j \varrho_j$. ($\varrho_j$ may not necessarily be a set of physical states.) By finding $\mathcal{E}(\varrho_j)$ through the action of the quantum process on the probe states and exploiting the linearity of the process, the output state is given by $\mathcal{E}(\rho)=\sum_j c_j \mathcal{E}(\varrho_j)$. This procedure, in general, can be cumbersome since finding $\mathcal{E}(\varrho_j)$ may require probe states that are not easy to prepare in the laboratory, especially for large dimensional Hilbert spaces. Hence, as discussed, quantum process tomography with coherent states based on the Glauber-Sudarshan representation, $\mathcal{E}(\rho)\! =\! \int\dm{\alpha} P(\alpha)\, \mathcal{E}(\ket{\alpha}\!\bra{\alpha})$, is of great interest, as the coherent probe states are readily available in the laboratory. However, for many quantum states, such as squeezed states and cat states, the \GSP\  exists as a highly singular distribution, and hence, using this expression for the output state is not useful, in general. As discussed in Sec.~\ref{Klaud-filter}, this problem can be avoided by considering the Fock basis $\{\varrho_{n,m}=\vert m \rangle \langle n \vert\}$ and estimating $\mathcal{E}(\vert m \rangle \langle n \vert)$ from the effect of the process on coherent states, $\mathcal{E}(\ket{\alpha}\!\bra{\alpha})$~\cite{Lobino2008,Rahimi2011}. This approach is limited to a truncated Hilbert space in the Fock basis, and the Hilbert space truncation introduces an error in the estimation of the output state for input states with nonzero high-order photon-number components such as squeezed states and cat states~\cite{Rahimi2011}.

Our formalism provides a more direct approach for the output estimation because by employing a physical filtering procedure, we can directly use $\mathcal{E}(\ket{\alpha}\!\bra{\alpha})$ in the Glauber-Sudarshan representation to estimate the output state of an unknown quantum process. One can use a filter function, such as the nonclassicality filter, whose Fourier transform is a probability density and whose width can be adjusted by the parameter $L$, to approximate any arbitrary input state $\rho$ by a filtered one $\rho_\Omega$. Then by using Eq.~(\ref{eq:reg-P-representation}), the corresponding output state, 
\begin{equation}
	\mathcal{E}(\rho_{\Omega})=\!\int\!\dm{\alpha} P_{\Omega}(\alpha) \mathcal{E}(\ket{\alpha}\!\bra{\alpha}),
\end{equation} 
is an approximation of the actual output state $\mathcal{E}(\rho)$. To investigate the error associated with this approximation,  we note that for any CPTP map the trace distance between the two output states is upper bounded by the trace distance between the corresponding input states $D\left(\mathcal{E}(\rho_1),\mathcal{E}(\rho_2)\right)\leq D\left(\rho_1,\rho_2\right)$ \cite{NielsenChuang}. Therefore, by using Eqs.~(\ref{eq:tra-distance}) and (\ref{eq:ent-fild}) we can see that for any $\delta>0$ there exists an $L$ such that
\begin{equation}
D\left(\mathcal{E}(\rho),\mathcal{E}(\rho_{\Omega})\right) \leq \sqrt{1-F_e({\rho},{\mathcal{E}}_\Omega)} \leq \delta.
\end{equation}   
As a consequence, the output estimation error can be made arbitrarily small by adjusting the width of the filter function.

This formalism is particularly useful  when the action of the quantum process on all coherent states can be simply described.
As an example, consider a loss channel with transmissivity $\eta$ which transforms coherent states into coherent states with attenuated amplitudes $\mathcal{E}_{\text{Loss}}(\ket{\alpha}\!\bra{\alpha})=\ket{\eta\alpha}\!\bra{\eta\alpha}$. Hence, any quantum state under a loss channel can be approximated by
\begin{equation}
	\mathcal{E}_{\text{Loss}}(\rho_{\Omega})=\frac{1}{\eta^2}\!\int\!\dm{\alpha} P_{\Omega}\!\left({\alpha}/{\eta}\right)\ket{\alpha}\!\bra{\alpha},
\end{equation}
where $P_{\Omega}(\alpha)$ is the regularized \GSP\ of the input state $\rho$. One can similarly consider other examples such as squeezing and amplification channels whose action on the coherent states is known. 

\subsection{Estimating the outcome probabilities of measurements using heterodyne measurement}

In general, a quantum measurement is described by a positive operator-valued measure (POVM) $\{\Pi_n\}$, where measurement operators satisfy $\Pi_n\geq0$ and $\sum_n \Pi_n=\mathcal{I}$. Given a quantum system in state $\rho$, the probabilities of measurement outcomes are given by the Born rule, $p(n|\rho)=\Tr[\rho\, \Pi_n]$. If the state is unknown, however, one has to perform a set of informationally complete measurements on an ensemble of identically prepared copies of the state in order to estimate $\rho$ through a procedure known as quantum state tomography. In general, quantum state tomography requires many measurement settings that, particularly for systems with infinite-dimensional Hilbert space, are very challenging. A readily available, informationally complete measurement for optical modes is heterodyne whose POVM elements are proportional to coherent states $\ket{\alpha}\!\bra{\alpha}/\pi$ and with which one directly samples from the Husimi $Q$ function of the quantum state $Q(\alpha)=\bra{\alpha}\rho\ket{\alpha}/\pi$. By using this measurement, for instance, one can construct an estimate of the density matrix in the Fock basis by evaluating multiple derivatives of $Q(\alpha)$ at the origin~\cite{mandel_wolf_1995}. 

Another interesting application of our formalism is to estimate probabilities of outcomes of a quantum measurement, such as photon counting, for an unknown quantum state using heterodyne measurement. For a given measurement with POVM elements $\Pi_n$, whose \GSP s can be highly singular, the outcome probabilities $p(n|\rho)$ can be approximated by
 \begin{align}
 	\begin{split}
p(n|\rho_\Omega)&=\Tr\left[\mathcal{E}_\Omega(\rho)\, \Pi_n\right]=\Tr\left[\rho\, \mathcal{E}^{*}_\Omega(\Pi_n)\right]\\
&=\pi\! \int\!\dm{\alpha} P_{\Omega}(n|\alpha)\, Q(\alpha),
 	\end{split}
 \end{align}
where the dual of the filtering map, $\mathcal{E}^{*}_\Omega$ is equal to $\mathcal{E}_\Omega$, assuming $\tilde{\Omega}(-\alpha)=\tilde{\Omega}(\alpha)$.
In the above equation, we have used Eq.~(\ref{eq:filterng-map}) for the CPTP filtering map and $ P_{\Omega}(n|\alpha)$ is the regularized \GSP\ of the POVM element $\Pi_n$.  Hence, by using an appropriate physical filtering map, the integral of highly-singular distributions can be converted into a regular one which, by having samples of the Husimi $Q$ function from heterodyne measurement, can be estimated using standard techniques.

For a given measurement, the trace distance between quantum states~(\ref{eq:tra-distance}) is an upper bound on the trace distance between the corresponding probabilities~\cite{NielsenChuang}
\begin{equation}
\frac{1}{2}\sum_n\big\lvert p(n|\rho)-p(n|\rho_{\Omega})\big\lvert
\leq \sqrt{1-F_e({\rho},{\mathcal{E}}_\Omega)}\leq \delta.
\end{equation}  
Therefore, by adjusting the width parameter of the filter function, one can achieve a desired level of accuracy $\delta$ in the probability estimation.
   

\section{Conclusions}\label{sec:conclusion}
We have studied phase-space filter functions, which are particularly useful for regularizing the Glauber-Sudarshan $P$ function, in the context of quantum maps in the space of operators in an infinite-dimensional Hilbert space. We have shown that the necessary and sufficient condition for such a map to be a quantum process, meaning completely positive and trace preserving, is positive semidefiniteness of the filter function. This physicality condition guarantees that the output of the filtering map is always a physical density operator with a regular Glauber-Sudarshan $P$ function. Examples of such a filter function are Gaussian filters and nonclassicality filters, whose Fourier transforms are probability densities. Note that physical filtering maps not only can be used as a theoretical tool for approximating quantum states but can also be implemented and used in an experiment. For instance, following the experimental method in~\cite{Kuhn2018}, one can use physical filtering maps to transform any nonclassical state to another nonclassical state with a regular Glauber-Sudarshan $P$ function, which can be used as a resource in quantum information science. This state-preparation technique, in practice, can replace the handling of highly singular Glauber-Sudarshan $P$ functions by an experimental procedure~\cite{Kuhn2018}.

The physicality condition derived in this work enables us to use the standard distance measures in quantum information to compare the states before and after this filtering. Importantly, we derived a lower bound on the fidelity between the states before and after the filtering and showed that by adjusting the width of the filter function this distance can be made arbitrarily small. Hence, using this formalism, any nonclassical state with a highly singular Glauber-Sudarshan $P$ function can be approximated, to an arbitrary accuracy, by a quantum state with a regular Glauber-Sudarshan $P$ function. This result makes the Glauber-Sudarshan representation a more practical tool in quantum information processing. As an interesting application in quantum information science, we considered estimating the output state of an unknown quantum process by knowing its action on coherent states. This experimental method can be used for the tomography of quantum processes for continuous-variable systems, in setups similar to previous experiments~\cite{Lobino2008,LobinoPRL2009,KumarPRL2013,Cooper2015,Kervinen2023}. In addition, we have shown that physical filtering maps can also be used to estimate the output probabilities of any measurement by using samples from the Husimi $Q$ function. For an unknown quantum state, the samples can be obtained using a readily available heterodyne measurement. This application can be useful for verifying the outcome probabilities of quantum experiments.

We also studied Klauder's filter functions, which were originally proposed to regularize the Glauber-Sudarshan $P$ functions. We showed that the associated filtering map does not meet the physicality requirement derived in this work. This implies that the output of Klauder's filtering map may be unphysical and using the corresponding regularized Glauber-Sudarshan $P$ functions may lead to unphysical results, such as negative probabilities. In this case, the output and input of Klauder's filtering maps cannot be compared using the fidelity or the trace distance, which have useful operational meanings in quantum information~\cite{NielsenChuang}. Hence, care must be taken in employing Klauder's filter functions in applications such as quantum process tomography~\cite{Lobino2008}, and further steps are required to impose the physicality, which can lead to an additional error.

Moreover, we considered filter functions that are the characteristic functions of density operators and have been considered in the context of the Wigner-function smoothing procedure. By using our criterion, we then identified a class of filter functions whose associated maps are positive but not completely positive. This class of filtering maps can be useful for witnessing the entanglement of quantum states~\cite{Horodecki-RMP-2009}. We leave this as a subject for future research. 

Another open question is to investigate the connection between filtering maps and measures of nonclassicality. Gaussian filtering has been used to define a measure of nonclassicality~\cite{Lee1991}. This measure can be viewed as the minimum amount of Gaussian noise required to make the filtered Glauber-Sudarshan $P$ function non-negative. By using the bound on the fidelity derived in this work, one can study whether the minimum Gaussian noise corresponds to the minimum distance between some states and the filtered ones with non-negative Glauber-Sudarshan $P$ function, or whether by using a non-Gaussian filter function the distance between the two states can be further minimized.

\section*{ACKNOWLEDGMENTS} The authors would like to acknowledge useful discussions with F. Narcowich and H. Nha.


%

\end{document}